\documentclass[11pt]{amsart}
\usepackage{amsthm}
\usepackage{amssymb}
\usepackage{enumerate}
\usepackage{graphicx}
\usepackage{float}
\usepackage{bbm}
\usepackage{amsmath}
\usepackage{comment}
\usepackage{bm}
\usepackage{extarrows}
\usepackage{subfigure}
\usepackage[linkcolor=red,anchorcolor=blue,citecolor=green]{hyperref}
\usepackage{algorithm,algorithmic}
\usepackage{booktabs}

\newtheorem{corollary}{Corollary}[section]
\newtheorem{proposition}{Proposition}[section]
\newtheorem{theorem}{Theorem}[section]
\newtheorem{lemma}[theorem]{Lemma}

\theoremstyle{definition}
\newtheorem{definition}[theorem]{Definition}
\newtheorem{example}[theorem]{Example}

\theoremstyle{remark}
\newtheorem{remark}[theorem]{Remark}

\numberwithin{equation}{section}

\newcommand{\R}{\mathbb{R}}
\newcommand{\Z}{\mathbb{Z}}
\newcommand{\N}{\mathbb{N}}

\title[ ]{System identification in Dynamical Sampling}
\author{Sui Tang}
\address{Department of Mathematics}
\curraddr{Vanderbilt University, Nashville, TN, 37240}
\email{sui.tang@vanderbilt.edu}
\thanks{2010 Mathematics Subject Classification: Primary 94A20 , 94A12, 42C15, 15A29.\\
\indent Key words. Discrete Fourier Analysis, Distributed sampling, reconstruction, channel\\ \indent estimation.\\
 \indent The research of this work is supported by NSF Grant  DMS-1322099.}
\begin{document}
\maketitle
\begin{abstract}
We consider the problem of spatiotemporal sampling in a discrete infinite dimensional spatially invariant evolutionary process $x^{(n)}=A^nx$ to recover an unknown convolution operator $A$ given by a filter $a \in \ell^1(\Z)$ and an unknown initial state $x$ modeled as  avector in $\ell^2(\Z)$. Traditionally, under appropriate hypotheses, any $x$ can be recovered from its samples on $\Z$ and $A$ can be recovered by the classical techniques of deconvolution. In this paper, we will exploit the spatiotemporal correlation and propose a new spatiotemporal sampling scheme to recover $A$ and $x$ that allows to sample the evolving states $x,Ax, \cdots, A^{N-1}x$ on a sub-lattice of $\Z$, and thus achieve the spatiotemporal trade off. The spatiotemporal trade off is motivated by several industrial applications \cite{Lv09}. Specifically, we show that $\{x(m\Z), Ax(m\Z),\\ \cdots,  A^{N-1}x(m\Z): N \geq 2m\}$ contains enough information to recover a typical ``low pass filter" $a$ and $x$ almost surely, in which we generalize the idea of the finite dimensional case in  \cite{AK14}. In particular, we provide an algorithm based on a generalized Prony method for the case when both $a$ and $x$ are of finite impulse response and an upper bound of their support is known. We also perform the perturbation analysis based on the spectral properties of the operator $A$ and initial state $x$, and verify them by several numerical experiments. Finally, we provide several other numerical methods to stabilize the method and numerical example shows the improvement. \\

\noindent
\\

\end{abstract}

\section{Introduction}
\subsection{The dynamical sampling problem}
In situations of practical interest, physical systems evolve in time under the action of well studied operators, one common example is provided by the diffusion processes. Sampling of such an evolving system can be done by sensors or measurement devices that are placed at various locations and can be activated at different times. In practice, increasing the spatial sampling density is usually much more expensive than increasing the temporal sampling rate (\cite{Lv09}).  Given the different costs associated
with spatial and temporal sampling,  we aim to reconstruct any states in the evolutionary process using as few sensors as possible, but allow one to take samples at different time levels. Besides, in some cases, obtaining samples at a sufficient rate at any single time level may not even be possible, spatiotemporal sampling may resolve this issue by oversampling in time.  A natural question is whether one can compensate for insufficient spatial sampling densities by oversampling in time. This setting has departed from the classical sampling theory, where the samples are taken simultaneously at only one time level, see \cite{BAC12, QS08, P11, M11, RK12, JG12, AK01, DMQ09, QSy08}. Dynamical sampling is a newly proposed mathematical sampling framework. It involves studying the time-space patterns formed by the locations of the measurement devices and the times of their activation. Mathematically speaking, suppose $x$ is an initial distribution that is evolving in time satisfying the evolution rule: $$x_t=A_tx$$ where $\{A_t\}_{t \in [0, \infty)}$ is a family of evolution operators satisfying the condition $A_0=I$. Dynamical sampling asks the question: when do coarse samplings taken at varying times $\{ x|_{\Omega_0},(A_{t_1} x)|_{\Omega_1},\dots, (A_{t_N} x)|_{\Omega_N}\}$ contain the same information as a finer sampling taken at the earliest time?  One goal of dynamical sampling is to find all spatiotemporal sampling sets $(\chi,\tau)=\{\Omega_t, t \in \tau\}$ such that the certain classes of signals $x$ can be recovered from the spatiotemporal samples $x_t(\Omega_t), t\in \tau$. In the above cases, the evolution operators are assumed to be known. It has been well-studied in the context of various  evolutionary systems in a very general setting, see \cite{AUCS14, AS14, ADK13,AADP13, ADK12}. 

Another important problem arises when the evolution operators are themselves unknown or partially known. In this case, we are interested in finding all spatiotemporal sampling sets and certain classes of evolution operators so that the family $\{A_t\}_{t \in [0, \infty)}$  or their spectrum  and even the unknown states can be identified. We call such a problem the unsupervised system identification problem in dynamical sampling.

Applications to Wireless Sensor Networks (WSN) is a natural setting for Dynamical sampling. 
WSN are widely used in many industries, including the health, military, and environmental
industries, see \cite{WSN} for numerious examples. In WSN, a huge amount of sensors are 
distributed to monitor a physical field such as the pollution, temperature or pressure. The goal is to 
exploit the evolutionary structure and the placement of sensors to reconstruct an
unknown field. However, the current approaches and algorithms do not make use of 
the evolutionary structure of the field, see \cite{GGK12, GG09, GG09, GG092}. It is not always possible to place sampling devices
at the desired locations. Then we may be able to recover the desired information by
placing the sensors elsewhere and use the evolution process to recover the signals at
the relevant locations. In addition, dynamical sampling will make the reconstruction cheaper since we
use a reduced numer of sensors.

\subsection{Problem Statement} In this subsection, we state an instance of the unsupervised system identification problem of dynamical sampling in an infinite dimensional setting. Let $x \in \ell^2(\Z)$  be an unknown initial spatial signal and the evolution operator $A$ be given by an unknown convolution filter $a \in \ell^1(\Z)$ such that $Ax=a*x$. At time $t=n\in \N$, the signal $x$ evolves to be $x_n=A^nx=a^n*x$, where $a^n=a*a\cdots *a$.  We call this evolutionary system spatially invariant. Given the spatiotemporal samples with both $x$ and $A$ unknown, we would like to recover as much information about them as we can under the given various priors. Here we first study the case of uniform subsampling. Without loss of generality, we assume a positive odd integer $m$ ($m >1$) to be the uniform subsampling factor. At time level $t=l$, we uniformly undersample the evolving state $A^lx$ and get the spatiotemporal data
\begin{equation}\label{eq8.4}
y_l=(a^l*x)(m\Z),
\end{equation}
which is a sequence in $\ell^2(\Z)$.  It is obvious that at any single time level $t=l$, we can not determine the state $A^lx$ from the measurement $y_l.$  The problem we are going to consider can be summarized as follows:  
\vspace{2mm}

    \textit{Under what conditions on $a, m, N$ and  $x$, can $a$ and $x$ be recovered from the spatiotemporal samples $\{y_l: l=0,\cdots,N-1\}$, or equivalently, from the set of measurement sequences $\{x(m\Z),(a*x)(m\Z), \cdots, (a^{N-1}*x)(m\Z) \}$}? 
\vspace{2mm}

In \cite{AK14}, Aldroubi and Krishtal consider the recovery of an unknown $d \times d$ matrix $B$ and an unknown initial state $x \in \ell^2(\Z_d)$ from coarse spatial samples of its successive states $\{B^kx, k=0,1,\cdots\}$. Given an initial sampling set $\Omega \subset \Z_d=\{1,2,\cdots,d\}$,  they employ techniques related to Krylov subspace methods to show how large $l_i$ should be to recover all the eigenvalues of $B$ that can possibly be recovered from spatiotemporal samples $\{B^kx(i): i\in  \Omega,k=0,1,\cdots,l_i-1\}$. Our setup is very similar to the special case of regular invariant dynamical sampling problem in \cite{AK14}. In this special case, they employ a generalization of the well known Prony method that uses these regular undersampled spatiotemporal data first for the recovery of the filter $a$.  Then by using techniques developed in \cite{ADK12},  they show how to recover the initial state from these spatiotemporal samples. In this paper, we will address the infinite dimensional analog of this special case and provide more algorithms. In \cite{TG13}, Peter and Plonka use a generalized Prony method to reconstruct the sparse sums of the eigenfunctions of some known linear operators. Our generalization of Prony method shares some similar spirits with it, but deals with a fundamentally different problem.  In Sparse Fourier Transformation, see \cite{APML, MI10, MI13, DYA13}, the idea is to uniformly undersample the fixed signal with different factors so that one can group subsets of Fourier space together into a small number of bins to isolate frequencies, then take an Aliasing-Based Search by Chinese Remainder Theorem so that one can recover the coefficients and the frequencies. In our case, intuitively, one can think of recovering of the shape of an evolving wave by observing the  amplitude of its aliasing version at fixed coarse locations over a long period of time as opposed to acquiring all of the amplitudes at once, then by the given priors, one can achieve the perfect reconstructions.  Other similar work include the the Slepian-Wolf distributed source coding problem \cite{dscp} and the distributed sampling problem in \cite{HRLV10}.  Our problem, however, is very different from the above in the nature of the processes we study. Distributed sampling problem typically deals with two signals correlated by a transmission channel. We, on the other hand, can observe an evolution process at several instances and over longer periods of time. 

\subsection{Notations}

In the following, we use standard notations. By $\mathbb{N}$, we denote the set of all positive integers. The linear space of all column vectors with $M$ complex components is denoted by $\mathbb{C}^M$. The linear space of all complex $M \times N$ matrices is denoted by $\mathbb{C}^{M\times N}.$ For a matrix $\mathbf{A}_{M,N}=(a_{ij}) \in \mathbb{C}^{M\times N}$, its transpose is denoted by $\mathbf{A}^T_{M,N}$, its conjugate-transpose by $\mathbf{A}^*_{M,N},$ and its Moore-Penrose pseudoinverse by $\mathbf{A}^+_{M,N}.$ A square matrix $\mathbf{A}_{M,M}$ is abbreviated to $\mathbf{A}_{M}$. Its infinity norm is defined by  $$\lvert\lvert \mathbf{A}_{M} \rvert\vert_{\infty}=\max\limits_{1\leq i\leq M}(\sum\limits_{j=1}^M \lvert a_{ij} \rvert).$$ For a vector $\mathbf{z}=(z_i) \in \mathbb{C}^M$, the $M \times M$ diagonal matrix built from $\mathbf{z}$ is denoted by $diag(\mathbf{z})$. We define the infinity norm  $\lvert\lvert \mathbf{z} \rvert\rvert_{\infty}=\max\limits_{i=1,\cdots,M} \lvert z_i \rvert$.  It is easy to see that 
$$\lvert\lvert \mathbf{A}_{M} \rvert\vert_{\infty}=\max\limits_{\mathbf{z}\in \mathbb{C}^M, \lvert\lvert \mathbf{z} \rvert\rvert_{\infty}=1} \lvert\lvert  \mathbf{A}_{M} \mathbf{z} \rvert\rvert_{\infty}. $$ Further we use the known submatrix notation coincides with MATLAB. For example, $\mathbf{A}_{M,M+1}(1 : M, 2 : M + 1)$ is the submatrix of $\mathbf{A}_{M,M+1}$ obtained by extracting rows 1 through $M$ and columns 2 through $M + 1$, and $\mathbf{A}_{M,M+1}(1 : M, M + 1)$ means the last column vector of $\mathbf{A}_{M,M+1}$.

\begin{definition} The minimal annihilating polynomial of a square matrix $\mathbf{A}_M$ is $p^{\mathbf{A}_M}[z]$, if it is the monic polynomial of smallest degree among all the monic polynomials $p$
such that $p(\mathbf{A}_M) = 0$. We will denote the degree of $p^{\mathbf{A}_M}[z]$ by $deg(p^{\mathbf{A}_M})$. 
\end{definition}

Let  the monic polynomial $p[z]=\sum\limits_{k=0}^{M-1}p_kz^k+z^M$, the companion matrix of $p[z]$ is defined by $$\mathbf{C}^{p[z]}=\left( \begin{array}{ccccc} 
    0&0 & \cdots&0&-p_0\\
     1& 0 &\cdots&0&-p_1 \\
   0& 1 & \cdots&0 &-p_2\\
\vdots&\vdots&\ddots&\vdots&\vdots\\
 0& 0& \cdots&1& -p_{M-1}
   \end{array}
   \right).  $$

\begin{definition} Let $w_1,w_2,\cdots, w_n$ be $n$ distinct complex numbers, denote $\mathbf{w}=[w_1,\cdots,w_n]^T$,  the $n \times N$ Vandermonde matrix generated by $\mathbf{w}$ is defined by 
\begin{equation}\label{V}
\mathbf{V}_{n,N}(\mathbf{w})=\left( \begin{array}{cccc} 
    1 &w_1 & \cdots& w_1^{N-1}\\
     1& w_2 &\cdots& w_2^{N-1} \\
   \vdots& \vdots & \ddots & \vdots\\
  1& w_n& \cdots& w_n^{N-1}
   \end{array}
   \right).\\
\end{equation}  

\end{definition}

\begin{definition} For a sequence $c=(c_n)_{n \in \Z} \in \ell^1(\Z) \text{ or } \ell^2(\Z)$, we define its Fourier transformation to be the function on the Torus $\mathbb{T}=[0,1)$
$$\hat c(\xi)=\sum_{n\in \Z} c_n e^{-2\pi i n \xi}, \xi \in \mathbb{T}.$$
\end{definition}


The remainder of the paper is organized as follows: In section \ref{noisefree}, we discuss the noise free case. From a theoretical aspect, we show that we can reconstruct a ``typical low pass filter'' $a$ and the initial signal $x$ from the dynamical spatiotemporal samples $\{y_l\}_{l=0}^{N-1}$ almost surely, provided $N \geq 2m$. For the case when both $a$ and $x$ are of finite impulse response and an upper bound of their support is known, we propose a Generalized Prony Method algorithm to recover the Fourier spectrum of $a$. In section \ref{stabilityanalysis}, we provide a perturbation analysis of this algorithm. The estimation results are formulated in the rigid $\ell^{\infty}$ norm and give us an idea of how the performance depends on the system parameters $a,x$ and $m$. In section \ref{ne},  we do several numerical experiments to verify some estimation results. In section \ref{othernumericalmethods}, we propose several other algorithms such as Generalized Matrix Pencil method, Generalized ESPRIT Method  and Cadzow Denoising methods to improve the effectiveness and robustness of recovery. The comparison between algorithms is illustrated by a numerical example in section \ref{numericalexamples}.  Finally,  we summarize the work in section \ref{conclusion}.

\section{Noise-free recovery}
\label{noisefree}We consider the recovery of a frequently encountered case in applications when the filter $a \in \ell^1(\Z)$ is a ``typical low pass filter" so that $\hat a(\xi)$ is  real, symmetric and strictly decreasing on $ [0,\frac{1}{2}]$. An example of such a typical low pass filter is shown in Figure 1. The symmetry reflects the fact that there is often no preferential direction for physical kernels and monotonicity is a reflection of energy dissipation. 
\begin{figure}[htbp]
\label{lp}
 \centering
    \includegraphics[width=0.5\textwidth]{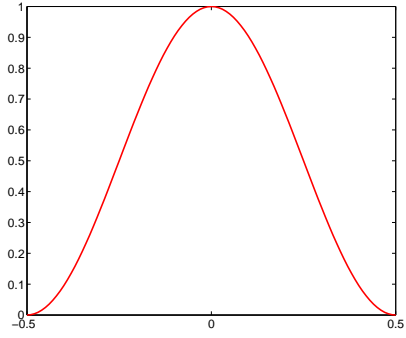}
 \caption{A Typical Low Pass Filter}
\end{figure}Without loss of generality, we also assume $a$ is a normalized filter, i.e., $\lvert \hat a(\xi) \rvert \leq 1, \hat a(0)=1$. In this section, we assume the spatiotemporal data $y_l=(a^l*x)(m\Z)$ is exact. Define the downsampling operator $S_m: \ell^2(\Z) \rightarrow \ell^2(\Z)$ by 
$$(S_m x)(k)=x(mk), k \in \Z,$$
then $y_l=S_m(a^l*x).$ Due to the Poisson Summation formula and the convolution theorem, we have the Lemma below for the downsampling operator. 
\begin{lemma}\label{PSF}The Fourier transform of each measurement sequence ${y}_l=S_m(a^l*x)$ at $\xi  \in \mathbb{T}$ is 
\begin{equation} \label{equmatrixprod}
 \hat y_l(\xi)=\frac{1}{m}\sum\limits_{ i=0 }^{m-1}   \hat{a}^l (\frac{\xi   + i}{m})\hat{x}(\frac{\xi  + i}{m}).   \end{equation}\label{timeseries}
\end{lemma}Let $N$ be an integer satisfying $N \geq 2m$, we define the $(N-m)\times 1$ column vector
\begin{equation}
\mathbf{h}_{t}(\xi)= [\hat {y}_t(\xi), \hat {y}_{t+1}(\xi),\cdots, \hat {y}_{N-m+t-1}(\xi)  ]^{T},
\end{equation}
and build the Hankel matrices

 \begin{equation}\label{hankelM}
\mathbf{H}_{N-m,m}^\xi(0)=\big[ \mathbf{h}_0(\xi), \mathbf{h}_1(\xi), \cdots, \mathbf{h}_{m-1}(\xi) \big],
\end{equation}
\[\mathbf{H}_{N-m,m}^\xi(1)=\big[ \mathbf{h}_1(\xi), \mathbf{h}_1(\xi), \cdots, \mathbf{h}_{m}(\xi) \big]. \]

For $\xi \in \mathbb{T}$, we introduce the notations $\mathbf{x}(\xi)=[\hat x(\frac{\xi}{m}),\cdots,\hat x(\frac{\xi+m-1}{m})]^T$ and $\mathbf{w(\xi)}=[\hat a(\frac{\xi}{m}),\cdots, \hat a(\frac{\xi+m-1}{m})]^T$. 

\begin{proposition}\label{hankelrecovery}
Let  $N$ be an integer satisfying $N \geq 2m$.

\begin{enumerate}
\item Then the rectangular Hankel matrices can be factorized in the following form:
\begin{equation}\label{decompose}
m\mathbf{H}_{N-m,m}^\xi(s)=\mathbf{V}^T_{m,N-m}(\mathbf{w}(\xi))diag(\mathbf{x}(\xi))diag(\mathbf{w(\xi)})^s\mathbf{V}_{m}(\mathbf{w}(\xi)),
\end{equation}where $s=0,1$. The Vandermonde matrix $\mathbf{V}_{m,N-m}(\mathbf{w}(\xi))$ and $\mathbf{V}_{m}(\mathbf{w}(\xi))$ are given in the way as indicated in Definition \ref{V}. 
\vspace{2mm}

\item Assume the entries of $\mathbf{x}(\xi)$ are all nonzero. The rank of the Hankel matrix $\mathbf{H}_{N-m,m}^\xi(0)$ can be summarized as follows:
\[
\text{Rank }\mathbf{H}_{N-m,m}^\xi(0)=
 \begin{cases}
m &\mbox{ if $\xi \neq 0 \text{ or } \frac{1}{2}$},\\
\frac{m+1}{2}&\mbox{otherwise}.
 \end{cases} \]

\item Assume the entries of $\mathbf{x}(\xi)$ are all nonzero. For $\xi \neq 0,\frac{1}{2}$, the vector $\mathbf{q}(\xi)=[q_0(\xi),\cdots,q_{m-1}(\xi)]^T$ is the unique solution of the linear system
\begin{equation}\label{pe}
\mathbf{H}_{N-m,m}^{\xi}(0)\mathbf{q}(\xi)=-\mathbf{h}_m{(\xi)}
\end{equation}if and only if the polynomial
\begin{equation}
q^{\xi}[z]=\sum\limits_{k=0}^{m-1}q_k(\xi)z^k+z^m 
\end{equation}with  coefficients given by $\mathbf{q}(\xi)$ is the minimal  annihilating polynomial of
the diagonal matrix $diag(\mathbf{w}(\xi))$.  In other words, the polynomial $q^{\xi}[z]$ has all  $\hat a(\frac{\xi+i}{m})\in \R\text{ }( i=0,\cdots,m-1)$ as roots.  Moreover,  if $p[z]$ is a monic polynomial of degree $m$, then 
\begin{equation}\label{companion}
\mathbf{H}_{N-m,m}^\xi(0)\mathbf{C}^{p[z]}=\mathbf{H}_{N-m,m}^\xi(1)
\end{equation}
if and only if $p[z]$ is the minimal annihilating polynomial of $diag(\mathbf{w}(\xi))$. 

\end{enumerate}

\end{proposition}

\begin{proof}(1) By Lemma \ref{PSF}, for $t=0,\cdots,m-1$, we have the identity:

  $$m\mathbf{h}_t(\xi)=\mathbf{V}_{m,N-m}^T(\mathbf{w}(\xi))diag(\mathbf{x}(\xi))\mathbf{V}_{m}(\mathbf{w}(\xi))(:,t+1).$$
Hence the first identity follows by the definition of $\mathbf{H}_{N-m,m}^\xi(0)$. Notice that for $t\geq 1$, 
$$m \mathbf{h}_t(\xi)=\mathbf{V}_{m,N-m}^T(\mathbf{w}(\xi))diag(\mathbf{x}(\xi))diag(\mathbf{w}(\xi))\mathbf{V}_{m}(\mathbf{w}(\xi))(:,t),$$ the second identity follows similarly. \\
(2) By the symmetric and monotonicity condition of $\hat a$ on $\mathbb{T}$, we have
\begin{equation}
\text{Rank } V_{m}(\mathbf{w}(\xi))=
 \begin{cases}
m &\mbox{ if $\xi \neq 0 \text{ or } \frac{1}{2}$},\\
\frac{m+1}{2}&\mbox{otherwise}.
 \end{cases} 
\end{equation}
Since $N \geq 2m$, $\text{Rank } \mathbf{V}_{m}(\mathbf{w}(\xi))=\text{Rank } \mathbf{V}^T_{m,N-m}(\mathbf{w}(\xi))$. By our assumptions, $diag(\mathbf{x}(\xi))$ is invertible. The rank of Hankel matrix $\mathbf{H}_{N-m,m}^\xi(0)$ can be computed by its factorization results in (1).\\
(3) If $\xi \neq 0 \text{ or }\frac{1}{2}$, then the diagonal matrix $diag(\mathbf{w}(\xi))$ has $m$ distinct eigenvalues consist of $\{\hat a(\frac{\xi+i}{m}): i=0,\cdots,m-1\}.$ The minimal annihilating polynomial of $diag(\mathbf{w}(\xi))$ is of degree $m$. Suppose $q^{\xi}[z]=\sum\limits_{k=0}^{m-1}q_k(\xi)z^k+z^m$ is the minimal annihilating polynomial of $diag(\mathbf{w}(\xi))$, $q^{\xi}[diag(\mathbf{w}(\xi))]=0.$ In other words,
$$\sum\limits_{k=0}^{m-1}q_k(\xi)diag(\mathbf{w}(\xi))^k=-diag(\mathbf{w}(\xi))^m.$$ Then

\begin{equation}\label{comphank}
\begin{split}
\mathbf{H}_{N-m,m}^{\xi}(0)\mathbf{q}(\xi)&=\sum\limits_{k=0}^{m-1}q_k(\xi)\mathbf{h}_k(\xi)\\ &=\mathbf{V}^T_{m, N-m}(\mathbf{w}(\xi)) (\sum\limits_{k=0}^{m-1}q_k(\xi)diag(\mathbf{w}(\xi))^k )\mathbf{x}(\xi)\\&=-\mathbf{V}^T_{m, N-m}(\mathbf{w}(\xi))diag(\mathbf{w}(\xi))^m\mathbf{x}(\xi)\\&=-\mathbf{h}_m(\xi).
\end{split}
\end{equation}
Conversely,  if $\mathbf{q}(\xi)$ is the solution of linear system \eqref{pe},  let the monic polynomial given by $ \mathbf{q}(\xi)$ be $q^{\xi}[z]$, then by the computation process of \eqref{comphank}, we have
$$\mathbf{V}^T_{m,N-m}(\mathbf{w}(\xi)) q^{\xi}[diag(\mathbf{w}(\xi))] \mathbf{x}(\xi)=0.$$
Since $\mathbf{V}^T_{m, N-m}(\mathbf{w}(\xi))$ is full column rank,  $q^{\xi}[diag(\mathbf{w}(\xi))] \mathbf{x}(\xi)=0.$ By the fact that $q^{\xi}[diag(\mathbf{w}(\xi))]$ is diagonal and $\mathbf{x}(\xi)$ has no zero entries, we know $q^{\xi}[z]$ is a monic annihilating polynomial of $diag(\mathbf{w}(\xi))$. The minimality is followed by counting its degree. If $p[z]$ is a monic annihilating polynomial of $diag(\mathbf{w}(\xi))$, by  computations, it is easy to show the identity \eqref{companion} is an equivalent formulation with the identity \eqref{pe}.
\end{proof}

\begin{corollary}In the case of $\xi=0$ or $\frac{1}{2}$, if $diag(\mathbf{x}(\xi))$ is invertible, then the coefficient vector of the minimal annihilating polynomial of $diag(\mathbf{w}(\xi))$ $\mathbf{c}(\xi) \in \R^{\frac{m+1}{2}}$  is the unique solution of the following linear system:
\begin{equation}\label{pronye}
\mathbf{H}_{N-m,\frac{m+1}{2}}^{\xi}(0)\mathbf{c}(\xi)=-\mathbf{h}_{\frac{m+1}{2}}{(\xi)},
\end{equation} where $\mathbf{H}_{N-m,\frac{m+1}{2}}^{\xi}(0)=\big [\mathbf{h}_0(\xi),\cdots,\mathbf{h}_{\frac{m-1}{2}}(\xi) \big].$
\end{corollary}

Let $\mu$ denote the Lebesgue measure on $\mathbb{T}$, and $X$ be a subclass of $\ell^2(\Z)$ defined by 
$$X=\{x\in \ell^2(\Z): \mu(\{\xi \in \mathbb{T}:\hat x(\xi)=0\})=0\}.$$
Clearly, $X$ is a dense class of $\ell^2(\Z)$ under the norm topology. In noise free scenario,we show that we can recover $a$ and $x$ provided that our initial state $x \in X$.

\begin{theorem} 
\label{thm8.3}Let $x \in X$ be the initial state and the evolution operator $A$ be a convolution operator given by  $a\in \ell^1(\Z)$ so that $\hat a(\xi)$ is real, symmetric, and strictly decreasing on $[0,\frac{1}{2}]$. Then $a$ and $x$ can be recovered from the set of measurement sequences $\{y_l=(a^l*x)(m\Z): l=0,\cdots, N-1\}$ defined in $\eqref{eq8.4}$ when $N \geq 2m$.
\end{theorem}  
  
\begin{proof} 
Since Fourier transformation is an isometric isomorphism from $\ell^2(\Z)$ to $L^2(\mathbb{T})$, we can look at this recovery problem on the Fourier domain equivalently. We are going to show that the regular subsampled data $\{y_l\}_{l=0}^{N-1}$ contains enough information to recover the Fourier spectrum of $a$ on $\mathbb{T}$ up to a measure zero set. By our assumptions of $x$, there exists a measurable subset $E_0$ of $\mathbb{T}$ with $\mu(E_0)=1$, so that $diag(\mathbf{x}(\xi))$ is an invertible matrix for $\xi \in E_0$.  Let $E=E_0-\{0,\frac{1}{2}\}$, if $\xi \in E$, by (3) of Proposition \ref{hankelrecovery}, we can recover the minimal  annihilating polynomial of $diag(\mathbf{w}(\xi))$. Now to recover the diagonal entries of $diag(\mathbf{w}(\xi))$, it amounts to finding the roots of this minimal annihilating polynomial and ordering them according to the monotonicity and symmetric condition on $\hat a$. In summary,  for each $\xi \in E$, we can uniquely determine $\{\hat a(\frac{\xi+i}{m}):i=0,\cdots,m-1\}$. Note $\mu(E)=1$,  and hence we can recover the Fourier spectrum of $a$ up to a measure zero set.  The conclusion is followed by applying the inverse Fourier transformation on $\hat a(\xi)$. Once $a$ is recovered, we can recover $x$ from the spatiotemporal samples $\{y_l\}_{l=0}^{m-1}$ using techniques developed in \cite{ADK13}.
\end{proof}

Theorem \ref{thm8.3} addresses the infinite dimensional analog of Theorem 4.1 in \cite{AK14}. If we don't know anything about $a$ in advance, with minor modifications of the above proof, one can show the recovery of  the range of $\hat a$ on a measurable subset of $\mathbb{T}$, where the measure of this subset is 1.

\begin{definition} Let $a=(a(n))_{n \in \Z}$, the support set of $a$ is defined by $Supp(a)=\{k\in \Z: a(k)\neq 0\}$. If $Supp(a)$ is a finite set,  $a$ is said to be of finite impulse response. 
\end{definition}

In particular, if $x$ is of finite impulse response, then $x \in X$. Now if both $x$ and $a$ are of finite impulse response, and we know an upper bound $r \in \N$ such that 
$Supp(a)$ and $Supp(x)$ are contained in $\{-r, -r+1, \cdots,r\}$, then we can compute the value of the Fourier 
transformation of $\{y_l\}_{l=0}^{N-1}$ at any $\xi \in \mathbb{T}.$ From the proof of Theorem \ref{thm8.3}, we can give an algorithm
similar to the classical Prony method to recover $\{\hat a(\frac{\xi+i}{m}):i=0,\cdots,m-1\}$ almost surely, given $\xi$ chosen uniformly from $\mathbb{T}$. 
It is summarized in Algorithm $\ref{algo:prony}.$


\begin{algorithm}
  \caption{Generalized Prony Method }\label{algo:prony}
  \begin{algorithmic}[1]
    \REQUIRE$ N\geq 2m$, $r \in \N$, $ \{y_l\}_{l=0}^{N-1}$, $\xi (\neq 0,\frac{1}{2}) \in \mathbb{T}.$
    
    \STATE Compute the Fourier transformation of the measurement sequences $\{y_l\}_{l=0}^{N-1}$ and build the Hankel matrix $\mathbf{H}^{\xi}_{N-m,m}(0)$ and the vector $\mathbf{h}_{m}(\xi)$.

\STATE  Compute the solution of the overdetermined linear system \eqref{pe}:
\[
\mathbf{H}_{N-m,m}^{\xi}(0)\mathbf{q}(\xi)=-\mathbf{h}_m{(\xi)}.
\] Form the polynomial $q^{\xi}[z]=\sum\limits_{k=0}^{m-1}q_k(\xi)z^k+z^m$ and find its roots, this can be done by solving the standard eigenvalue problem of its companion matrix. 

\STATE Order the roots by the monotonicity and symmetric condition of $\hat a$ to get $\{\hat a(\frac{\xi+i}{m}): i=0,\cdots,m-1\}$.\\
\ENSURE $\{\hat a(\frac{\xi+i}{m}): i=0, \cdots,m-1\}.$
\end{algorithmic}
\end{algorithm}


\begin{corollary} 
\label{cor1} In addition to the assumptions of Theorem \ref{thm8.3}, if both $a$ and $x$ are of finite impulse response  with support contained in $\{-r, -r + 1, \cdots , r\}$ for some $r \in \mathbb{N}$,  then it is enough to determine $a$ and $x$ after we recover $\{ \hat a(\eta_i) : i = 1, \cdots, r\}$ at $r$ distinct locations by Algorithm \ref{algo:prony}.
\end{corollary}

\begin{proof}
Under these assumptions, we know 
\begin{equation}
\hat a (\xi)=a(0)+\sum_{k=1}^{r} a(k) \cos(2\pi k \xi).
\end{equation}  
Suppose $\{ \hat a(\eta_i): i=1,\cdots, r, \eta_i \neq \eta_j \text{ if $i \neq j$} \}$ are recovered,  we set up the following linear equation
 \begin{equation}\label{eq4}
  \left( \begin{array}{ccccc} 
    1 & cos(2\pi \eta_1) &\cdots& cos(2r\pi \eta_1)\\
   1& cos(2\pi \eta_2)&\cdots& cos(2r\pi \eta_2)\\
  \vdots& \vdots & \cdots  & \vdots\\
1& cos(2\pi \eta_{r}) &\cdots&cos(2r \pi \eta_{r})
   \end{array}
   \right) \left( \begin{array}{c} 
   a(0)\\
   a(1)\\
  \vdots\\
   a(r) \end{array} \right)=\left( \begin{array}{c} 
   \hat a(\eta_1)\\
   \hat a(\eta_2)\\
  \vdots\\
  \hat a(\eta_{r}) \end{array} \right).
\end{equation}
Note that $\{1,\cos(2\pi \eta ), \cdots,\cos(2r\pi \eta )\}$ is a Chebyshev system on $[0,1]$(see \cite{AFT63}),   and hence \eqref{eq4} has a unique solution.
Then we can recover $x$ by solving the linear system 
$$\mathbf{V}^T_{m, N-m}(\mathbf{w}(\xi)) \mathbf{x}(\xi)=\mathbf{h}_0(\xi)$$
for finitely many $\xi$s, which finishes the proof.
\end{proof}

\section{Perturbation Analysis} 
\label{stabilityanalysis}
In previous sections, we have shown that if we are able to compute the spectral data $\{\hat y_l(\xi)\}_{l=0}^{N-1}$ at $\xi$, then we can recover the Fourier spectrum $\{\hat a(\frac{\xi+i}{m}):i=0,\cdots,m-1\}$ by Algorithm \ref{algo:prony}. However, we assume the spectral data are noise free. A critical issue still remains. We need to analyze the accuracy of the solution achieved by Algorithm \ref{algo:prony} in the presence of noise. Mathematically speaking, assume the measurements are given by $\{\tilde{y}_l\}_{l=0}^{N-1}$ compared to $\eqref{eq8.4}$ so that $\lvert\lvert \hat y_l(\xi) -\widehat{\tilde{y}_l}(\xi)\rvert\rvert_{\infty} \leq \epsilon_l$ for all $\xi \in \mathbb{T}$. Given an estimation for $\epsilon = max_l  \lvert \epsilon_l \rvert$, how large can the error be in the worst case for the output parameters of Algorithm \ref{algo:prony} in terms of $\epsilon$, and the system parameters $a, x$ and $m$.  Most importantly, we need to understand analytically what kind of effects that the subsampling factor $m$ will impose on the performance of the Algorithm \ref{algo:prony}. 

In this section, for simplicity, we choose $N=2m$ to meet the minimal requirement. In this case, the Hankel matrix $\mathbf{H}^{\xi}_{N-m,m}(0)$ is a square matrix and the vectors $\mathbf{h}_t(\xi)$ are of length $m$. We denote them by two new notations:  $\mathbf{H}_m(\xi)$ and $\mathbf{b}_t(\xi)$.  Our perturbation analysis will consist of two steps. Suppose our measurements are perturbed from $\{y_l\}_{l=0}^{2m-1}$ to $\{\tilde{y}_l\}_{l=0}^{2m-1}$. For any $\xi$, we firstly measure the perturbation of $\mathbf{q}(\xi)$ in terms of $\ell^{\infty}$ norm. Secondly we measure the perturbation of the roots.  It is well known that the roots of a polynomial are continuously dependent on the small change of its coefficients, see Proposition \ref{polyp}. Hence, for a small perturbation,  although the roots of the perturbed polynomial $\tilde{q}^{\xi}[z]$ may not be real, we can order them according to their modulus and have a one to one correspondence with the roots of ${q}^{\xi}[z]$. Before presenting our main results in this section, let us introduce some useful notations and terminologies.


\begin{definition} \label{def6}Let $\xi \in \mathbb{T}-\{0,\frac{1}{2}\}$, consider the set $\{\hat a(\frac{\xi+i}{m}):i=0,\cdots,m-1 \}$ that consists of m distinct nodes. 
\begin{enumerate}
\item For $0 \leq k \leq m-1$, the separation between $\hat a(\frac{\xi+k}{m})$ with other $m-1$ nodes is measured  by $$ \delta_k(\xi)=\frac{1}{\prod\limits_{j \neq k \atop 0\leq j\leq m-1 } \lvert \hat a(\frac{\xi+j}{m})-\hat a(\frac{\xi+k}{m}) \rvert}. $$

\item For $0 \leq k \leq m$, the $k$-th elementary symmetric function generated by the $m$ nodes is denoted by 
\begin{equation}
\sigma_{k}(\xi)=
 \begin{cases}
1 &\mbox{ if $k$=0},\\
\sum\limits_{0\leq j_1<\cdots<j_k\leq m-1}  \hat a(\frac{\xi+j_1}{m})\hat a(\frac{\xi+j_2}{m})\cdots \hat a(\frac{\xi+j_k}{m})&\mbox{otherwise}.
 \end{cases}
\end{equation}
 For $0\leq k, i \leq m-1$, the $k$-th elementary symmetric function  generated by  $m-1$ nodes  with $\hat a(\frac{\xi+i}{m})$ missing is denoted by $\sigma_{k}^{(i)}(\xi)$.
 \end{enumerate}
\end{definition}

The following Proposition measures the perturbation of the polynomial roots in terms of the perturbation of its coefficients and is the key to our perturbation analysis. 

\begin{proposition}[see Proposition $\mathbf{V}.1$ in \cite{RBG06}]\label{polyp}
Let $z_k$ be a root of multiplicity $M_k \in \mathbb{N}^+$ of the $r$-th order polynomial $p[z].$ For all $\epsilon >0$, let $p_{\epsilon}[z]=p[z]+\epsilon \Delta p[z],$ where $\Delta p[z]$ is a polynomial of order lower than $r$. Suppose that $\Delta p[z_k] \neq 0.$ Then there exists a positive $\epsilon_0$ such that for all $\epsilon <\epsilon_0$ there are exactly $M_k$ roots of $p_{\epsilon}[z]$, denoted $\{z_{k,m}(\epsilon)\}_{m \in \{0,\cdots, M_k-1\}}$, which admit the first-order fractional expansion
\begin{equation}
z_{k,m}(\epsilon)=z_k+\epsilon^{\frac{1}{M_k}}\Delta z_k e^{2\pi i \frac{m}{M_k}}+O(\epsilon^{\frac{2}{M_k}}),
\end{equation} where $\Delta z_k$ is an arbitrary $M_k$-th root of the complex number 

\begin{equation}
(\Delta z_k)^{M_k}=-\frac{\Delta p[z_k]}{\frac{1}{M_k!} p^{(M_k)}[z_k]}.
\end{equation}

\end{proposition}

\begin{proposition}\label{prop3}Let the perturbed measurements $\{\tilde{y}_l\}_{l=0}^{2m-1}$ be given with an error satisfying $\lvert\lvert \widehat{\tilde{y}_l}(\xi)-\hat{y_l}(\xi) \rvert\rvert_{\infty}\leq \epsilon, \forall l$. Let $\widetilde{\mathbf{H}}_m(\xi)$ and $\tilde{\mathbf{b}}_m(\xi)$ be given by $\{\widehat{\tilde{y}_l}(\xi)\}_{l=0}^{2m-1}$ in the same way as in \eqref{hankelM} and \eqref{timeseries}. Assume $\mathbf{H}_m(\xi)$ is invertible and $\epsilon$ is sufficient small so that $\widetilde{\mathbf{H}}_m(\xi)$ is also invertible. Denote by $ \tilde {\mathbf{q}}(\xi)$ the solution of  the linear system $\widetilde{\mathbf{H}}_m(\xi)\tilde {\mathbf{q}}(\xi)=- \tilde{\mathbf{b}}_m(\xi)$. Let $\tilde{q}^{\xi}[z]$ be the Prony polynomial formed by $\tilde {\mathbf{q}}(\xi)$ and $\{\tilde{ \hat {a}}(\frac{\xi+i}{m}):i=0,\cdots,m-1 \}$ be its roots, then we have the following estimates as $\epsilon \rightarrow 0$, 
\begin{equation}
\label{estimation1}
\lvert \lvert  \mathbf{q}(\xi)-\tilde{\mathbf{q}}(\xi) \rvert\rvert_{\infty}  \leq \lvert \lvert \mathbf{H}_m^{-1}(\xi) \rvert\rvert_{\infty}(1+m \beta_1(\xi)  )\epsilon+O(\epsilon^2),
\end{equation}
where $\beta_1(\xi)=\max\limits_{k=1,\cdots,m}\lvert \sigma_{k}(\xi) \rvert$.  As a result, we achieve the following first order estimation
\begin{equation}
\label{estimation2}
\lvert \tilde{\hat{a}}(\frac{\xi+i}{m})-\hat a(\frac{\xi+i}{m})\rvert \leq C_i(\xi)(1+m\beta_1(\xi))\lvert \lvert \mathbf{H}_m^{-1}(\xi) \rvert\rvert_{\infty}\epsilon+O(\epsilon^2),
\end{equation}
where $C_i(\xi)= \delta_i(\xi) \cdot (\sum\limits_{k=0}^{m-1} \lvert \hat  a ^{k}(\frac{\xi+i}{m}) \rvert) $. 
\end{proposition}

\begin{proof} Note that linear system $\eqref{pe}$ is perturbed to be 
\begin{equation}
\widetilde{\mathbf{H}}_m(\xi) \tilde{\pmb{q}}(\xi)=-\tilde{\mathbf{b}}_m(\xi).
\end{equation}
By our assumptions, we have 
\begin{equation}
\lvert\lvert \Delta{\mathbf{H}_m}(\xi) \rvert\rvert_{\infty}=\lvert\lvert \widetilde{\mathbf{H}}_m(\xi)-\mathbf{H}_m(\xi) \rvert\rvert_{\infty} \leq m \epsilon,
\end{equation}
\begin{equation}
\label{estimation1}
\lvert\lvert \Delta{\mathbf{b}}_m(\xi) \rvert\rvert_{\infty}=\lvert\lvert \tilde {\mathbf{b}}_m(\xi)-\mathbf{b}_m(\xi) \rvert\rvert_{\infty} \leq  \epsilon.
\end{equation}
Define $\Delta \mathbf{q}(\xi)=\tilde{\mathbf{q}}(\xi)-\mathbf{q}(\xi)$, by simple computation, 
\begin{equation}
\label{esq}
\Delta \mathbf{q}(\xi)=\mathbf{H}_m^{-1}(\xi)(I+\mathbf{H}_m^{-1}(\xi) \Delta \mathbf{H}_m(\xi))^{-1}(-\Delta \mathbf{b}_m(\xi)-\Delta \mathbf{H}_m(\xi) \mathbf{q}(\xi)). 
\end{equation}
Hence if $\epsilon \rightarrow 0$, we obtain
\begin{equation}
\label{esqr}
\Delta \mathbf{q}(\xi)=\mathbf{H}_m^{-1}(\xi)(-\Delta \mathbf{b}_m(\xi)-\Delta \mathbf{H}_m(\xi)\mathbf{q}(\xi))+O(\epsilon^2). 
\end{equation}
Now we can easily get an estimation of $\ell^{\infty}$ norm of $ \Delta \mathbf{q}(\xi)$
\begin{equation}
\label{espoly}
\lvert \lvert \Delta \mathbf{q}(\xi) \rvert\rvert_{\infty} \leq \lvert \lvert \mathbf{H}_m^{-1}(\xi) \rvert\rvert_{\infty}(1+m \lvert \lvert \mathbf{q}(\xi) \rvert\rvert_{\infty})\epsilon+O(\epsilon^2). 
\end{equation}
Since $\{\hat a(\frac{\xi+i}{m}): i=0,\cdots,m-1 \}$ are the roots of $q^{\xi}[z]$, using 
Vieta's Formulas(see \cite{V03}), we know $$\lvert\lvert  \mathbf{q}(\xi) \rvert\rvert_{\infty} = \mathop{max}_{1 \leq k\leq m} \lvert \sigma_{k}(\xi) \rvert. $$
Let $(\Delta q(\xi))[z]$ be the polynomial of degree less than or equal to $m-1$ defined by the vector $\Delta \mathbf{q}(\xi)$. Using Proposition \ref{polyp}, and denote by $(q^{\xi})^{'}[z]$ the derivative function of $q^{\xi}[z]$, for $0 \leq i \leq m-1$, we conclude 
\begin{equation}
\label{esroots}
\begin{split}
 \lvert \tilde{\hat{a}}(\frac{\xi+i}{m})-\hat a(\frac{\xi+i}{m})\rvert &=\lvert \frac{\Delta \mathbf{q}(\xi)[\hat a(\frac{\xi+i}{m})]}{ (q^{\xi})^{'}[\hat a(\frac{\xi+i}{m})]}+O(\epsilon^2) \lvert \\
& \leq \frac{\lvert\lvert \Delta \mathbf{q}(\xi) \rvert\rvert_{\infty} (\sum\limits_{k=0}^{m-1}\lvert \hat a ^{k}(\frac{\xi+i}{m})\rvert)}{\prod\limits_{j \neq i \atop 0\leq j\leq m-1} \lvert \hat a(\frac{\xi+j}{m})-\hat a(\frac{\xi+i}{m}) \rvert}+O(\epsilon^2)\\
&\leq   C_i(\xi)\lvert \lvert \mathbf{H}_m^{-1}(\xi) \rvert\rvert_{\infty}(1+m\mathop{max}_{1 \leq k\leq m}\lvert \sigma_{k}(\xi)\rvert)\epsilon+O(\epsilon^2),
\end{split}
\end{equation}
where $C_i(\xi)=\delta_i(\xi) (\sum\limits_{k=0}^{m-1} \lvert \hat a ^{k}(\frac{\xi+i}{m})\rvert)$.

\end{proof}

Therefore it is important to understand the relation between the behavior of $\lvert\lvert \mathbf{H}_m^{-1}(\xi) \rvert\rvert_{\infty}$ and our system parameters, i.e, $a$, $m$ and $x$. Next, we are going to estimate $\lvert\lvert \mathbf{H}_m^{-1}(\xi) \rvert\rvert_{\infty}$ and reveal their connection with the spectral properties of $a$,$x$ and the subsampling factor $m$.

\begin{proposition}\label{theorem3} Assume $\mathbf{H}_m(\xi)$ is invertible,  we have the lower bound estimation
\begin{equation}
\label{hle}
  \lvert\lvert \mathbf{H}_m^{-1}(\xi) \rvert\rvert_{\infty}  \geq m \cdot \max_{i=0,\cdots,m-1}  \frac{\beta_2(i,\xi) \delta_i(\xi) }{\lvert \hat x (\frac{\xi+i}{m})\rvert}, 
\end{equation}
where $\beta_2(i,\xi)=\max\limits_{k=0,\cdots,m-1} \lvert \sigma_k^{(i)}(\xi) \rvert$, and the upper bound estimation

\begin{equation}
\label{hue}
\lvert\lvert \mathbf{H}_m^{-1}(\xi) \rvert\rvert_{\infty} \leq m \cdot \max_{i=0,\cdots,m-1} \frac{(\delta_i(\xi) \prod\limits_{j \neq i \atop 0\leq j\leq m-1} (1+ \lvert \hat a(\frac{\xi+j}{m}) \rvert))^2 }{\lvert \hat x (\frac{\xi+i}{m})\rvert}. 
\end{equation}
\end{proposition}

\begin{proof}
Firstly, we prove the lower bound for $\lvert\lvert \mathbf{H}_m^{-1}(\xi) \rvert\rvert_{\infty}$. Denote the Vandermonde matrix $\mathbf{V}_m(\mathbf{w}(\xi))$ by the  abbreviated $\mathbf{V}_m(\xi)$. Suppose $\mathbf{V}_m^{-1}(\xi)=(v_{ki})_{1\leq k,i\leq m}$ is the inverse of $\mathbf{V}_m(\xi)$, by the inverse formula for a standard Vandermonde matrix,
$$v_{ki}=(-1)^{m-k} \sigma_{m-k}^{(i-1)}(\xi) \delta_{i-1}(\xi).$$
Let $\{e_i\}_{i=1}^{m}$ be the standard basis for $\mathbb{C}^m$ and $w_i(\xi)=\mathbf{V}_m^T(\xi)e_i$ for $i=1,\cdots,m$. Since $\lvert \hat a(\xi) \rvert \leq 1$, we conclude that $\lvert\lvert w_i \rvert\rvert_{\infty} =1$.  
\begin{equation}
\label{lwes}
\begin{split}
\lvert\lvert \mathbf{H}_m^{-1}(\xi)\rvert\rvert_{\infty} &\geq \max_{i=1,\cdots,m} \lvert\lvert \mathbf{H}_m^{-1}(\xi) w_i(\xi) \rvert\rvert_{\infty} \\ &\geq m\cdot \max_{i=1,\cdots,m} \frac{\lvert\lvert \mathbf{V}_m^{-1}(\xi) e_i \rvert\rvert_{\infty}}{\lvert \hat x (\frac{\xi+i}{m})\rvert}  \\\
& = m\cdot \max_{i=0,\cdots,m-1 }  \frac{ \beta_2(i,\xi)\delta_i(\xi)}{\lvert \hat x (\frac{\xi+i}{m})\rvert}.
\end{split}
\end{equation}
On the other hand, using the factorization \eqref{decompose} and the upper bound norm estimation for the inverse of a Vandermonde matrix in \cite{GW75}, we show that 
\begin{equation}
\label{hmupper}
\begin{split}
\lvert \lvert \mathbf{H}_m^{-1}(\xi) \rvert\rvert_{\infty} & \leq m\lvert \lvert \mathbf{V}_m^{-1}(\xi) \rvert\rvert_{\infty} \lvert \lvert ((\mathbf{V}_m^{-1})^{T}(\xi)) \rvert\rvert_{\infty} \lvert \lvert diag^{-1}(\mathbf{x}(\xi)) \rvert\rvert_{\infty}\\
& \leq  m \max_{i=0,\cdots,m-1} \frac{(\delta_i(\xi) \prod\limits_{j \neq i \atop 0\leq j\leq m-1} (1+ \lvert \hat a(\frac{\xi+j}{m}) \rvert))^2 }{\lvert \hat x (\frac{\xi+i}{m})\rvert}. 
\end{split}
\end{equation}
\end{proof}

As an application of Proposition $\ref{theorem3}$, the following theorem sheds some light on the dependence of $\lvert \lvert \mathbf{H}_m^{-1}(\xi) \rvert\rvert_{\infty}$ on m.

\begin{theorem}\label{cor2} If $\lvert \hat x (\xi) \rvert \leq  M$ for every $\xi \in \mathbb{T}$, then $\lvert \lvert \mathbf{H}_m^{-1}(\xi) \rvert\rvert_{\infty} \geq O(2^m)$. Therefore, $\lvert \lvert \mathbf{H}_m^{-1}(\xi) \rvert\rvert_{\infty}\rightarrow \infty$ as $m \rightarrow \infty$.
\end{theorem}
\begin{proof}
We show this by proving $m \cdot \max\limits_{i=0,\cdots,m-1} \delta_i(\xi) \geq O(2^m)$.  Note $\beta_2(i,x)\geq \lvert \sigma_0^{(i)}(\xi)\rvert =1$. By \eqref{lwes}, 
\begin{equation}\label{mev}
\lvert\lvert \mathbf{H}_m^{-1}(\xi) \rvert\rvert_{\infty}  \geq m \cdot \frac{ \max\limits_{i=0,\cdots,m-1} \delta_i(\xi) }{M}=O(2^m),
\end{equation}
the conclusion follows. 
Let $c(\xi)= \max\limits_{i=0,\cdots,m-1} \delta_i(\xi)$. Note that 
\begin{equation}
\label{maxestimation}
\begin{split}
\frac{1}{c(\xi)^m} \leq \prod\limits_{i=0}^{m-1} \frac{1}{\delta_i(\xi)}&=\prod\limits_{0\leq i<j\leq m-1} \lvert  \hat a(\frac{\xi+i}{m})-\hat a(\frac{\xi+j}{m})\rvert^2\\ &= \lvert det(\mathbf{V}_m(\xi) \rvert^2. 
\end{split}
\end{equation}
Since every entry of $\mathbf{w}(\xi)$ is contained in $[-1,1]$, the Chebyshev points on $[-1,1]$ maximize the determinant of Vandermonde matrix, see \cite{SDE14}. Therefore,  by the formula for the determinant of a Vandermonde matrix on the Chebyshev points in \cite{AGP98},  we get 
$$ \lvert \det(\mathbf{V}_m(\xi))\rvert^2 \leq \frac{m^m}{2^{(m-1)^2}}. $$
By $\eqref{maxestimation}$,
$$ c(\xi) \geq \frac{2^{\frac{(m-1)^2}{m}}} {m}$$

which implies that  $m\cdot c(\xi) \geq O(2^m).$
Hence by $\eqref{mev}$
$$ \lvert\lvert \mathbf{H}_m^{-1}(\xi) \rvert\rvert_{\infty} \geq O(2^m) \rightarrow \infty, m\rightarrow \infty. $$ 
\end{proof}
\begin{remark} By our proof, we also see  that  $\lvert \lvert \mathbf{H}_m^{-1}(\xi) \rvert\rvert_{\infty}$ grows at least geometrically when m increases. 
\end{remark}
Summarizing, our results in this section suggest that 
\begin{enumerate}

\item For $0\leq k \leq m-1$, the accuracy of recovering the node $\hat a(\frac{\xi+k}{m})$ not only depends on its separation with other nodes $\delta_k(\xi)$(see Definition \ref{def6}), but also depends on the global minimal  separation $\delta(\xi)=\max\limits_{k=0,\cdots,m-1} \delta_k(\xi)$ among the nodes. Fix $m,x$, our estimations \eqref{esroots} and \eqref{hmupper} suggest that error  $\lvert \Delta_k(\xi) \rvert= \lvert \hat{\tilde{a}}(\frac{\xi+k}{m})-\hat a(\frac{\xi+k}{m})\rvert$ in the worst possible case could be proportional to $\delta_k(\xi)\delta^2(\xi)$. Our numerical experiment suggests this is sharp, see Figure 2 (c) and (d). 
\item  The accuracy of recovering all nodes is inversely proportional to the lowest magnitude of $\{\hat x(\frac{\xi+i}{m}): i=0,\cdots,m-1\}$.
\item Increasing $m$ may result in amplifying the error caused by the noise significantly. Since by the proof of Theorem $\ref{cor2}$, $ \lvert\lvert \mathbf{H}_m^{-1} \rvert\rvert_{\infty}$ grows at least geometrically when $m$ increases. Thus,  when m increases, the infinity norm of $ \mathbf{H}^{-1}_m(\xi)$ gets bigger and our solutions become more likely less robust to noise, see Figure 2 (a) and (b).
\end{enumerate}

\section{Numerical Experiment}
\label{ne}
In this section, we provide some simple numerical simulations to verify some theoretical accuracy estimations in section \ref{stabilityanalysis}. 

\subsection{Experiment Setup}
Suppose our filter $a$ is around the center of radius 3. For example, Let $a=(\cdots 0, 0.05, 0.4,0.1, 0.4, 0.05,0,\cdots)$ such that $\hat a(\xi)=0.1+0.8\cos(2 \pi \xi)+0.1\cos(4\pi\xi)$, $x=(\cdots,0, 0.242, 0.383, 0.242, 0, \cdots)$ such that  $\hat x(\xi)=0.383+0.484\cos(2\pi\xi).$ We choose $m=3$.
\begin{enumerate}
 \item  In this experiment, we choose 9 points $[\xi_1,\cdots,\xi_9]=0.49:0.001:0.498$ and calculate $\hat y_l(\xi_i)$ and  the perturbed $\widehat{\tilde{y}_l}(\xi_i)=\hat y_l(\xi_i)+\epsilon_l$ for $l=0,\cdots,5$, where $y_l$ is defined as in $\eqref{eq8.4}$ and $\epsilon_l \sim 10^{-10}$. 

\item Use Algorithm \ref{algo:prony} to calculate the roots of $q^{\xi}[z]$ and the perturbed roots of $\tilde{q}^{\xi}[z]$ respectively, then compute $\lvert \Delta_k(\xi_i) \rvert= \lvert \tilde{\hat{a}}(\frac{\xi_i+k}{m})-\hat a(\frac{\xi_i+k}{m})\rvert$ for $k=0,1,2.$

\item  Choose $\xi=0.3$ and $m=2:1:7$, we compute $\lvert\lvert \mathbf{H}^{-1}_m(0.3)\rvert\rvert_{\infty}$  for different m. 
\end{enumerate}

\subsection{Experiment Results}
\vspace{1mm}
In this subsection, we plot several figures to reflect the experiment results. The $x$-axis of the Figure 2 (a)$-$(e) are set to be 1:9, which represent $\xi_1,\cdots,\xi_9.$
\begin{enumerate}

\item \textbf{The dependence of $\max\limits_{k}\lvert \Delta_k(\xi) \rvert$ on the infinity norm of $\mathbf{H_m^{-1}(\xi)}$}.
Since the points $\xi_1,\cdots,\xi_9$ are more and more closer to $\frac{1}{2}$, we expect the infinity norm of $\mathbf{H}^{-1}_m(\xi)$ to get sufficiently larger and larger.  Note that $m$ and $x$ are fixed, the quantity $\mathbf{H}^{-1}_m(\xi)$ is the only significantly large item in the error estimations. We plot the value of $\lvert\lvert \mathbf{H}^{-1}_m(\xi_i)\rvert\rvert_{\infty}$  and $\max\limits_{k}\lvert \Delta_k(\xi_i) \rvert$ for $i=1,\cdots,9$ in Figure 2 (a) and (b). They exhibit almost the same behaviour and grows proportionally.  This indicates that the bigger $\lvert\lvert \mathbf{H}^{-1}_m(\xi_i)\rvert\rvert_{\infty}$ is, the bigger  $\max\limits_{k}\lvert \Delta_k(\xi_i) \rvert$ is.

\item \textbf{Sharpness of estimation \eqref{estimation2} and \eqref{hue}}. Our estimation \eqref{estimation2} and \eqref{hue} suggest that error  $\lvert \Delta_k(\xi) \rvert$ in the worst possible case could be proportional to $\delta_k(\xi)\delta^2(\xi)$. We plot the value of $\lvert \Delta_2(\xi_i) \rvert$ and $\delta_2(\xi)\delta^2(\xi)$
for $i=1,\cdots,9$ in Figure 2 (c) and (d). It is indicated that $\Delta_2(\xi_i)$ grows approximately proportionally to the growth of $\delta_2(\xi_i)\delta^2(\xi_i)$, which suggests the sharpness of estimation\eqref{estimation2} and \eqref{hue}. It is worthy to mention that the curve of $\max\limits_{k}\lvert \Delta_k(\xi_i) \rvert$ coincides with the curve of $\lvert \Delta_2(\xi_i) \rvert$, and the curve of $\max\limits_k \delta_k(\xi_i)\delta^2(\xi_i)$ coincides with the curve of $\delta_2(\xi_i)\delta^2(\xi_i)$. Since in this experiment, $m$ and $x$ are fixed, this also suggests that the quantity $\delta_k(\xi_i)\delta^2(\xi_i)$ essentially decides the accuracy. The bigger the quantity is, the less accuracy the Algorithm is.

\item \textbf{The infinity norm of $\mathbf{H}_m^{-1}(\xi)$}. Recall in this experiment, we choose $m=2,3,\cdots,6,7$ and $\xi=0.3$. We  plot the value of $\lvert\lvert \mathbf{H}_m^{-1}(0.3) \rvert\rvert_{\infty}$ for different $m$. The results are presented in Figure 2 (f).  The $y-$axis is set to be logarithmic. It is shown that $\lvert\lvert \mathbf{H}_m^{-1}(\xi) \rvert\rvert_{\infty}$ grows geometrically. 

\end{enumerate}

\begin{figure}[htbp]
\label{fig1}

  \centering
    \includegraphics[width=\textwidth]{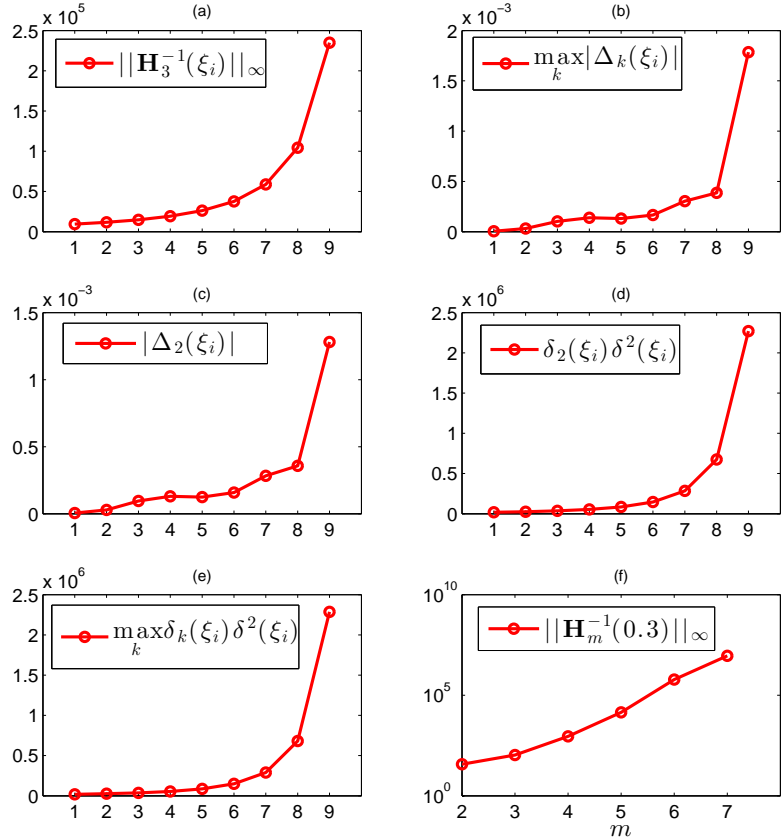}
  \caption{Experiment Results}
\end{figure}

\section{Other Numerical mehtods}
\label{othernumericalmethods}


In the following subsections,  we will investigate the data structure of the Hankel matrix built from the spatiotemporal samples and present two  algorithms based on the classical matrix pencil method and ESPRIT estimation method to our case. These two classical methods are well known for their better numerical stability than the original Prony method.

\subsection{Generalized Matrix Pencil Method}

Let $L$ and $N$ be two integers satisfying $L \geq m$ and $N \geq L+m$. Similarly, we define the $(N-L)\times 1$ column vector 
$$\mathbf{h}_{t}(\xi)= [\hat {y}_t(\xi), \hat {y}_{t+1}(\xi),\cdots, \hat {y}_{N-L+t-1}(\xi)  ]^{T} ,$$
and form the rectangular Hankel matrices 
\begin{equation}\label{hankelM1}
\mathbf{H}_{N-L,L+1}^{\xi}=\big[ \mathbf{h}_0(\xi), \mathbf{h}_1(\xi), \cdots, \mathbf{h}_{L}(\xi) \big],
\end{equation}
\[ \mathbf{H}_{N-L,L}^\xi(s)=\mathbf{H}_{N-L,L+1}^{\xi}(1:N-L, s+1:L+s), s=0, 1. \]
Similar to the case $L=m$, for $s=0,1,$
\begin{equation}\label{fac2}
\mathbf{H}_{N-L,L}^\xi(s)=V_{m,N-L}(\mathbf{w}(\xi))^T diag(\mathbf{x}(\xi))diag(\mathbf{w}(\xi))^sV_{m,L}(\mathbf{w}(\xi)).
\end{equation}

Recall that the superscripts ``$*$'' and ``$+$" will denote the conjugate transpose and the pseudoinverse. The following Lemma provides a foundation with the Generalized Matrix Pencil method.

\begin{lemma}
Let $N,L$ be two postive integers s.t. $ m \leq L \leq N-m$. Assume $\xi \neq 0,\frac{1}{2}$ and $diag(\mathbf{x}(\xi))$ is invertible. The solutions to the generalized singular eigenvalue problem :
\begin{equation}\label{matrixpencil}
(z\mathbf{H}_{N-L,L}^\xi(0)-\mathbf{H}_{N-L,L}^\xi(1))\mathbf{p}(\xi)=0
\end{equation}
subject to $\mathbf{p}(\xi) \in \mathbf{R}(\mathbf{H^{*}}_{N-L,L}^\xi(0))$,  which denotes the column space of $\mathbf{H^{*}}_{N-L,L}^\xi(0)$ are
\[ z_i=\hat a (\frac{\xi+i-1}{m})\]
\[\mathbf{p}(\xi)=\mathbf{p}_i(\xi)=\text{ i-th column of } \mathbf{V}_{m,L}^+(\mathbf{w}(\xi))\]
for $i=1,\cdots,m.$
\end{lemma}

\begin{proof} The proof can be done by the factorization $\eqref{fac2}$ and a similar manner with the proof of Theorem 2 in \cite{Hua}.
\end{proof}

\begin{proposition}\label{mtx}
Let $N,L$ be two postive integers s.t. $ m \leq L \leq N-m$. Assume $\xi \neq 0,\frac{1}{2}$ and $diag(\mathbf{x}(\xi))$ is invertible. The $L \times L$ matrix  $\mathbf{H^+}_{N-L,L}^\xi(0)\mathbf{H}_{N-L,L}^\xi(1)$ has  $\{\hat a(\frac{\xi+i}{m}), i=0, \cdots,m-1\}$ and $L-m$ zeros as eigenvalues.
\end{proposition}

\begin{proof}
Left multiplying \eqref{matrixpencil} by $\mathbf{H^+}_{N-L,L}^\xi(0)$, we have
\begin{equation}
\mathbf{H^+}_{N-L,L}^\xi(0)\mathbf{H}_{N-L,L}^\xi(1)\mathbf{p}_i(\xi)=z_i \mathbf{H^+}_{N-L,L}^\xi(0) \mathbf{H}_{N-L,L}^\xi(0)\mathbf{p}_i(\xi),
\end{equation}
By property of pseudoinverse, $\mathbf{H^+}_{N-L,L}^\xi(0)\mathbf{H}_{N-L,L}^\xi(0)$ is the orthogonal projection onto the $\mathbf{R}(\mathbf{H^{*}}_{N-L,L}^\xi(0))$. Since $\mathbf{p}_i(\xi) \in \mathbf{R}(\mathbf{H^{*}}_{N-L,L}^\xi(0))$, it is easy to see that  the set $\{\hat a(\frac{\xi+i}{m}): i=0,\cdots,m-1\}$ are $m$ eigenvalues of $\mathbf{H^+}_{N-L,L}^\xi(0)\mathbf{H}_{N-L,L}^\xi(1).$  Since the rank of $\mathbf{H^+}_{N-L,L}^\xi(0)\mathbf{H}_{N-L,L}^\xi(1)$ is $m \leq L$, $\mathbf{H^+}_{N-L,L}^\xi(0)\mathbf{H}_{N-L,L}^\xi(1)$ has $L-m$ zero eigenvalues.  
\end{proof}

It is immediate to see that one advantage of the matrix pencil method is the fact that there is no need to compute the coefficients of the minimal annihilating polynomial of $diag(\mathbf{w}(\xi))$. In this way, we just need to solve a standard eigenvalue problem of a square matrix  $\mathbf{H^+}_{N-L,L}^\xi(0)\mathbf{H}_{N-L,L}^\xi(1)$. In order to compute $\mathbf{H^+}_{N-L,L}^\xi(0)\mathbf{H}_{N-L,L}^\xi(1)$, inspired by idea of Algorithm 5 for SVD based Matrix Pencil Method in \cite{Hua1}, we can employ the Singular Value Decomposition(SVD) of the Hankel matrices.

\begin{lemma}
In addition to the conditions of Proposition \ref{mtx}, given the SVD of the Hankel matrix, 

$$\mathbf{H}_{N-L,L+1}^{\xi}=\mathbf{U}_{N-L}^{\xi}\mathbf{\Sigma}_{N-L,L+1}^{\xi}\mathbf{W}_{L+1}^{\xi},$$
then $$\mathbf{H^+}_{N-L,L}^\xi(0)\mathbf{H}_{N-L,L}^\xi(1)={\mathbf{W}^{\xi}}^+_{L+1}(1:m,1:L)\mathbf{W}^{\xi}_{L+1}(1:m,2:L+1).$$
\end{lemma}

\begin{proof} This can be shown by direct computations and noticing that $\mathbf{H}_{N-L,L+1}^{\xi}$ has only $m$ nonzero singular values.
\end{proof}
We summarize the Generalized Matrix Pencil Method in Algorithm \ref{algo:matrixpencil}. Note that the amount of computation required by Algorithm \ref{algo:matrixpencil} depends on the free parameter $L$. Numerical experiments show that the choice of $L$  greatly affects the noise sensitivity of the eigenvalues. In terms of the noise sensitivity and computation cost, the good choice for $L$ is between one third of $N$ and two thirds of $N$ \cite{Hua1}.  In our numerical example, we choose $L$ to be around one third of $N$. 
\begin{algorithm}
  \caption{Generalized Matrix Pencil Method (Based on SVD)}\label{algo:matrixpencil}
  \begin{algorithmic}[1]
    \REQUIRE$ m \leq L \leq N-m$, $r \in \N$, $ \{y_l\}_{l=0}^{N-1}$, $\xi (\neq 0,\frac{1}{2}) \in \mathbb{T}.$
    
    \STATE Compute the Fourier transformation of the measurement sequences $\{y_l\}_{l=0}^{N-1}$ and build the Hankel matrix $\mathbf{H}_{N-L,L+1}^{\xi}$ and compute its SVD
$$\mathbf{H}_{N-L,L+1}^{\xi}=\mathbf{U}_{N-L}^{\xi}\mathbf{\Sigma}_{N-L,L+1}^{\xi}\mathbf{W}_{L+1}^{\xi}.$$

\STATE Compute the eigenvalues of ${W^{\xi}}^+_{L+1}(1:m,1:L)W^{\xi}_{L+1}(1:m, 2:L+1)$.

\STATE Delete $L-m$ smallest values in modulus (zeros in noise free case) from the eigenvalues. Order the rest eigenvalues by the monotonicity and symmetric condition of $\hat a$ to get $\{\hat a(\frac{\xi+i}{m}): i=0,\cdots,m-1\}$.\\
\ENSURE $\{\hat a(\frac{\xi+i}{m}): i=0, \cdots,m-1\}.$
\end{algorithmic}
\end{algorithm}


\subsection{Generalized ESPRIT Method}Original ESPRIT Method relies on a particular property of  Vandermonde matrices known as the rotational invariance \cite{6}. By the factorization results \eqref{fac2}, we have seen that the Hankel data matrix $\mathbf{H}_{N-L,L+1}^{\xi}$ containing successive spatiotemporal data of the evolving states  is rank deficient and that its range space, known as the signal subspace, is spanned by Vandermonde matrix generated by $\{\hat a(\frac{\xi+i}{m}), i=0, \cdots,m-1\}$. Hence we can generalize the idea and present the generalized ESPRIT algorithm based on SVD for estimating the $\{\hat a(\frac{\xi+i}{m}): i=0, \cdots,m-1\}$ in our case.  We summarize it in Algorithm \ref{algo:esprit}
\begin{algorithm}
  \caption{Generalized ESPRIT Algorithm}\label{algo:esprit}
  \begin{algorithmic}[1]
    \REQUIRE$ m \leq L\leq N-m$, $r \in \N$, $\{y_l\}_{l=0}^{N-1}$, $\xi (\neq 0,\frac{1}{2}) \in \mathbb{T}.$
    
    \STATE Compute the Fourier transformation of the measurement sequences $\{y_l\}_{l=0}^{N-1}$ and form the Hankel matrix $\mathbf{H}^{\xi}_{N-L,L+1}.$

\STATE  Compute the SVD of  $\mathbf{H}^{\xi}_{N-L,L+1}=\mathbf{U}^{\xi}_{N-L}\mathbf{\Sigma}^{\xi}_{N-L,L+1}\mathbf{W}^{\xi}_{L+1}.$

\STATE Compute the $m \times m$ spectral matrix $\mathbf{\Phi}(\xi)$ by solving the linear system

$$\mathbf{U}^{\xi}_{N-L}(1:N-L-1,1:m)\mathbf{\Phi}(\xi) =\mathbf{U}^{\xi}_{N-L}(2:N-L,1:m)$$ and estimate the eigenvalues of $\mathbf{\Phi}(\xi)$. 

\STATE Order the eigenvalues by the monotonicity and symmetric condition of $\hat a$ to get $\{\hat a(\frac{\xi+i}{m}): i=0,\cdots,m-1\}$.\\

\ENSURE $\{\hat a(\frac{\xi+i}{m}): i=0, \cdots,m-1\}.$
\end{algorithmic}
\end{algorithm}

\subsection{Data Preprocessing Using Cadzow Denoising Method} 

It has been shown in the previous sections the Hankel matrix $\mathbf{H}_{N-L,L+1}^{\xi}(m \leq L \leq N-m)$ has two key properties in the noise free case under appropriate hypothesis:
\begin{enumerate}
\item It has rank $m$.
\item It is Toeplitz.
\end{enumerate}
In the noisy case, these two properties are not initially satisfied simultaneously. $\mathbf{H}_{N-L,L+1}^{\xi}$ is very sensitive to noise, numerical experiments show that even very small noise ($\sim10^{-10}$) will change  its rank dramatically. To further improve robustness, we use an iterative method devised by Cadzow \cite{Cadzow} to preprocess the noisy data and  guarantee to build a Hankel matrix with above two key properties. In our context, it can be summarized in Algorithm \ref{algo:cadzow}.


\begin{algorithm}

  \caption{Cadzow Iterative Denoising Method}\label{algo:cadzow}
  \begin{algorithmic}[1]
    \REQUIRE$ m \leq L\leq N-m$, $\{ \hat {\tilde{y}}_l(\xi)\}_{l=0}^{N-1}$, $\xi (\neq 0,\frac{1}{2}) \in \mathbb{T}.$ 
    
  \STATE Build the Hankel matrix $\mathbf{H}^{\xi}_{N-L,L+1}$ from $\{ \hat {\tilde{y}}_l(\xi)\}_{l=0}^{N-1}$
and preform the SVD.  Let $\lambda_1,\cdots,\lambda_K$ be its singular values, $K=\min\{N-L, L+1\}.$

\STATE  Set $\epsilon$ to be a small positive number.

  \WHILE{ $\frac{\lambda_{m+1}}{\lambda_m} \geq \epsilon$}

\STATE Enforce the rank $m$ of  $\mathbf{H}^{\xi}_{N-L,L+1}$ by setting the $K-m$ smallest singular values to zero.

\STATE Enforce the Toeplitz structure on $\mathbf{H}^{\xi}_{N-L,L+1}$ by averaging the entries along the diagonals. \ENDWHILE 

\STATE Extract the denoised  Fourier  data $\{ \hat {\tilde{y}}_l(\xi)\}_{l=0}^{N-1}$ from the first column and the last row of  $\mathbf{H}^{\xi}_{N-L,L+1}$

\ENSURE  Denoised  Fourier  data $\{ \hat {\tilde{y}}_l(\xi)\}_{l=0}^{N-1}$ and Hankel matrix $\mathbf{H}^{\xi}_{N-L,L+1}.$
\end{algorithmic}
\end{algorithm}

￼The procedure of Algorithm \ref{algo:cadzow} is guaranteed to converge to a matrix which exhibits the desired two key properties \cite{Cadzow}. The iterations stop whenever the ratio of the ￼$(m+1)$-th singular value to the ￼ $m$-th one, falls below a predetermined threshold. Since  Algorithm \ref{algo:prony} does not perform well when noise is big,  we can combine the Algorithm \ref{algo:prony} and Algorithm \ref{algo:cadzow} to recover the Fourier spectrum of $a$ and improve the performance.  In our numerical example, we choose $L=m$.


\section{numerical examples}
\label{numericalexamples}
In this section, we present a numerical example to illustrate the effectiveness and robustness of the proposed Algorithms. 

\begin{example} Let the filter $$a=(\cdots, 0, 0.25, 0.5, 0.25, 0,\cdots)$$ so that $\hat a(\xi)=0.5+0.5\cos(2\pi \xi)$. $\hat a$ is approximately Gaussian on $[-\frac{1}{2}, \frac{1}{2}].$ Let the initial signal $x$ be  a conjugate symmetric vector given by $x(0)=0.75, x(1)=\bar{x}(-1)=0.8976+0.4305i,$ and $x(2)=\bar{x}(-2)=0.9856-0.1682i$ so that $\hat x(\xi)=0.75+2Re((0.9856-0.1682i)e^{-4\pi i \xi}+(0.8976-0.4305i)e^{-2\pi i \xi})$. The subsampling factor $m$ is set to be 5. Given the Fourier data of the spatiotemporal samples $\{\hat y_l\}_{l=0}^{N-1}$, we add independent uniform distributed noise $\epsilon_l \sim U(-\epsilon, \epsilon)$ to the Fourier data $\hat y_l$ for $l=0,\cdots,N-1$.  Recall that $\lvert \Delta_k(\xi)\rvert=\lvert \tilde{\hat{a}}(\frac{\xi+k}{m})-\hat a(\frac{\xi+k}{m})\rvert$, we define the relative error $$ \mathbf{e}_k(\xi) = \frac{\lvert \Delta_k(\xi)\rvert}{\max\limits_{k} {\lvert  \hat a(\frac{\xi+k}{m})\rvert }}$$
for $k=0,1,m-1$. The best case error  is set to be $\mathbf{e}_{best}(\xi)=\min\limits_{k}\mathbf{e}_k(\xi)$ and the worst case error is set to be $\mathbf{e}_{worst}(\xi)=\max\limits_{k}\mathbf{e}_k(\xi).$ Besides, we define the mean square error $$\mathbf{MSE}^2(\xi)=\frac{\sum\limits_{k=0}^{m-1} \lvert \Delta_k(\xi) \rvert^2}{\sum\limits_{k=0}^{m-1} \lvert  \hat a(\frac{\xi+k}{m}) \rvert^2}.$$   Then we apply our Algorithm \ref{algo:prony}, Algorithms \ref{algo:prony}+Algorithm \ref{algo:cadzow}, Algorithm \ref{algo:matrixpencil} and Algorithm \ref{algo:esprit} to the case when $\epsilon=0.4$.  For several parameters $N$ and $L$, the resulting errors (average over 100 experiments) are presented in Table \ref{aggiungi}.  As the bound $\epsilon$ in the algorithms we use $10−^{-10}$.  It is shown in the table that increasing the temporal samples, i.e. $N$, will help reduce the error.  The new proposed algorithms have better performance than Algorithm \ref{algo:prony}, if given more spatiotemporal data.

\begin{table}[tp]%
\caption{Numerical Results}
\label {aggiungi}\centering %

\begin{tabular}{ ccccccc}
\toprule %
\textbf{Algorithm}&$m$&$N$&$L$& $\mathbf{e}_{best}$&$\mathbf{e}_{worst}$ &$\mathbf{MSE}$\\\hline
1&&&&&&\\
& 5 &10& 5&0.4e-07&0.19e-05&0.14e-05\\
 &5&15&5& 0.27e-08& 0.24e-06&0.17e-06\\
&5&20&5 &0.85e-09&0.18e-06&0.13e-06\\
&5&25&5& 0.46e-09& 0.17e-06&0.12e-06\\ \midrule
1+4& 5&10&5& 0.38e-07& 0.18e-05&0.13e-05\\
& 5 &15&5&0.13e-08 & 0.13e-06&0.09e-06\\
& 5&20&5 & 0.15e-09&0.58e-07 &0.41e-07\\ 
&5&25&5&0.49e-10&0.42e-07&0.29e-07\\ \midrule
2& 5 &15&5& 0.17e-08&0.16e-06&012e-06\\
& 5&20& 6&0.25e-09&0.74e-07&0.53e-07\\
&5&25&8&0.69e-10&0.47e-07&0.33e-07\\ \midrule
3& 5 &15&5&0.17e-08&0.16e-06&0.11e-06\\
& 5&20& 6&0.21e-09&0.66e-07&0.46e-07\\
&5&25&8&0.62e-10&0.45e-07&0.32e-07\\ \bottomrule

\end {tabular}
\end {table}

\end{example}


\section{Conclusion}
\label{conclusion}
In this paper, we investigate the conditions under which we can recover a typical low pass convolution filter $a \in \ell^1(\Z)$ and a vector $x \in \ell^2(\Z)$ from the combined regular subsampled version of the vector  $x, \cdots, A^{N-1}x$ defined in $\eqref{eq8.4}$, where $Ax=a*x$.  We show that if one doubles the amount of temporal samples needed in \cite{ADK13} to recover the signal propagated by a known filter, one can almost surely solve the problem even if the filter is unknown. We firstly propose an algorithm based on the classical Prony method to recover the finite impulse response filter and signal, if an upper bound for their support is known.  In particular, we have done a first order perturbation analysis and the estimates are formulated in very simple geometric terms involving Fourier spectral function of $a, x $ and $m$, shedding some light on the structure of the problem. We get a lower bound estimation for infinity norm of  $\mathbf{H}^{-1}_m(\xi)$ in terms of $m.$  Then we propose several other algorithms, which can make use of  more temporal samples and increase the robustness to noise. The potential applications includes the One-Chip Sensing: sensors inside chips for accurate measurements of voltages, currents, and temperature (e.g., avoid overheating any area of the chip), sources localization of an evolving state and time-space trade off(e.g., sound field acquisition using microphones)etc.


\section{Acknowledgment}
\label{Acknowledgment}
The author would like to thank Akram Aldroubi for his endless encouragements and very insightful comments in the process of creating this work. The author is also indebted to Ilya Kristhal and Yang Wang for their helpful discussions related to this work.

\end{document}